\documentclass[11pt]{article}

\input{Style.sty}

\title{
Optimal Circuit Size for Fixed-Hamming-Weight Quantum States Preparation
}
\author{Jingquan Luo\thanks{luojq25@mail2.sysu.edu.cn} }
\author{Lvzhou Li\thanks{Corresponding author: lilvzh@mail.sysu.edu.cn}}
\affil{Institute of Quantum Computing  and Software, School of Computer Science and Engineering, Sun Yat-sen University, Guangzhou 510006, China}
\date{\today }

\begin{document}
\maketitle

\begin{abstract}
We study the problem of efficiently preparing fixed-Hamming-weight (HW-$k$) quantum states, which are superpositions of $n$-qubit computational basis states with exactly $k$ ones. 
We present a quantum circuit construction that prepares any $n$-qubit HW-$k$ state  with a circuit size of $\bo{\binom{n}{k}}$ using at most $\max\{0, n-3\}$ ancillary qubits. 
This is the first construction that  achieves the theoretical lower bound on circuit size while using only a small number of ancillary qubits. 
We believe that the techniques presented in this work can be extended to other quantum state preparation algorithms based on decision diagrams, 
potentially reducing the reliance on ancillary qubits or lowering the overall circuit size.
\end{abstract}

\section{Introduction}

Quantum state preparation (QSP) is a fundamental problem in quantum computing, as the ability to efficiently initialize quantum registers into desired states is a prerequisite for executing a wide range of quantum algorithms, such as Hamiltonian simulation \cite{childs2018toward, low2017optimal, low2019hamiltonian, berry2015simulating}, quantum machine learning \cite{lloyd2014quantum, kerenidis2016quantum, kerenidis2019q, kerenidis2021quantum, rebentrost2014quantum} and quantum linear system solvers~\cite{harrow2009quantum, childs2017quantum}. Since arbitrary state preparation requires a number of operations that scales exponentially with the number of qubits~\cite{shende2005synthesis, bergholm2005quantum, plesch2011quantum, sun2023asymptotically, zhang2022quantum, yuan2023optimal, zhang2024circuit, zhang2024practical, gui2024spacetime, low2024trading, gosset2024quantum}, identifying classes of states that admit efficient preparation methods has become an important direction in both theory and practice.

Over the past years, several structured families of quantum states have been investigated. 
For example, extensive research has been devoted to the preparation of sparse quantum states, in which only a limited number of computational basis states have nonzero amplitudes~\cite{gleinig2021efficient, malvetti2021quantum, ramacciotti2023simple, mozafari2022efficient, de2020circuit, de2022double, mao2024towards, luo2024circuit, farias2025quantum, sun2023asymptotically, zhang2022quantum, zhang2024circuit, luo2025space}. 
Dicke states, which are superpositions of all $n$-qubit basis states with fixed Hamming weight $k$, can be prepared by circuits of size $\mathcal{O}(kn)$~\cite{bartschi2019deterministic, bartschi2022short, aktar2022divide, yuan2025depth}. 
In addition, some works represent quantum states using decision diagrams and subsequently construct quantum circuits from these representations~\cite{mozafari2022efficient, mozafari2020automatic, tanaka2024quantum, hong2025advancing}. 
These results highlight that exploiting the inherent structure of a state can drastically lower the cost of preparation compared to generic decomposition techniques.

In this work, we focus on a particular family of structured quantum states: those whose nonzero amplitudes are restricted to computational basis states with a fixed Hamming weight. Equivalently, every basis state in the support of such a state has exactly $k$ ones in its binary representation. Formally, for a binary string $x \in \{0,1\}^n$, let $\mathrm{HW}(x)$ denote its Hamming weight, i.e., the number of ones in $x$. Then an HW-$k$ state on $n$ qubits is any quantum state of the form
\begin{align}
\ket{\psi} \;=\; \sum_{\substack{x \in \{0,1\}^n \\ \mathrm{HW}(x)=k}} \alpha_x \ket{x},
\end{align}
where $\alpha_x \in \mathbb{C}$ are arbitrary complex amplitudes satisfying the normalization condition $\sum_{x} |\alpha_x|^2 = 1$.
These states occupy a combinatorial subspace of dimension $\binom{n}{k}$ within the $2^n$-dimensional Hilbert space. 
HW-$k$ states can be viewed as a generalization of Dicke states: when all nonzero amplitudes are equal, an HW-$k$ state reduces to the standard $n$-qubit Dicke state with Hamming weight $k$.

For example, for $n=3$ and $k=2$, an HW-2 state has the general form
\begin{align}
\ket{\psi} \;=\; \alpha_{110}\ket{110} + \alpha_{101}\ket{101} + \alpha_{011}\ket{011}.
\end{align}

HW-$k$ states have found a variety of applications across quantum computing. In quantum machine learning, HW-$k$ states are used to mitigate issues such as barren plateaus~\cite{monbroussou2025trainability, kerenidis2022quantum, anselmetti2021local}. Similarly, in quantum finance, encoding strategies that restrict operations to HW-$k$ subspaces have been proposed to efficiently represent portfolio~\cite{he2023alignment}. 
From a physics perspective, HW-$k$ states are referred to $U(1)$-eigenstates, a class of Bethe states that includes exact eigenstates of spin-$1/2$ XXZ chains with open or periodic boundary conditions~\cite{raveh2024deterministic,van2022preparing,ruiz2024bethe,van2021preparing,sopena2022algebraic,li2022bethe}. Preparing these states on a quantum computer allows the computation of correlation functions and the simulation of strongly correlated quantum systems~\cite{somma2002simulating, wecker2015solving}. Thus, HW-$k$ states are both practically relevant and theoretically intriguing from the perspective of circuit synthesis and complexity analysis.

Several approaches have been proposed for preparing HW-$k$ states in the literature~\cite{mao2024towards, farias2025quantum, raveh2024deterministic, van2022preparing, ruiz2024bethe, van2021preparing, sopena2022algebraic, li2022bethe}. 
Refs.~\cite{farias2025quantum, raveh2024deterministic} independently proposed algorithms achieving a circuit size of $\bo{\binom{n}{k}k}$ without ancillary qubits, 
whereas Ref.~\cite{mao2024towards} achieved a circuit size of $\bo{\binom{n}{k}\log n}$ with $\bo{n}$ ancillary qubits. 
Refs.~\cite{van2022preparing, ruiz2024bethe, van2021preparing, sopena2022algebraic, li2022bethe} focused on preparing Bethe states through the coordinate Bethe ansatz~\cite{bethe1997theory, gaudin1971boundary} or the algebraic Bethe ansatz~\cite{faddeev1984spectrum, korepin1997quantum, sklyanin1988boundary}, and their methods do not generalize to general HW-$k$ states.  
Thus, the current state-of-the-art circuit size for HW-$k$ state preparation is $\bo{\min{\cbra{\binom{n}{k}k, \binom{n}{k}\log n}}}$. Regarding lower bounds, a simple dimensionality argument yields a trivial lower bound of $\Omega(\binom{n}{k})$. 
As noted at the end of Ref.~\cite{mao2024towards}, it remains an open problem whether HW-$k$ state preparation can achieve this bound.


\subsection{Contributions}
In this work, we aim to optimize the circuit size (i.e., the gate-count) of preparing HW-$k$ states. 
Our main result is given as follows and  the related
results is summarized in \cref{tab:summary}.

\begin{theorem}\label{thm:hwk}
    For any integers $n$ and $k$, any $n$-qubit HW-$k$ state can be prepared by a quantum circuit of size $\bo{\binom{n}{k}}$ using $\max\cbra{0, n-3}$ ancillary qubits.
\end{theorem}

This result establishes, for the first time, that HW-$k$ states can be prepared with a circuit size matching the theoretical lower bound of $\om{\binom{n}{k}}$, while using only a small number of ancillary qubits. Specifically, we improve upon the results of Refs.~\cite{farias2025quantum, raveh2024deterministic} by a factor of $k$, and upon the result of Ref.~\cite{mao2024towards} by a factor of $\log n$.  In addition, it is noted that our method is completely different from those in Refs. \cite{mao2024towards, raveh2024deterministic, farias2025quantum}, which enriches the toolbox of HW-$k$ state preparation.

It is worth mentioning that the circuit constructed by our algorithm only uses $X$ gates, CNOT gates, Toffoli gates, and single- or controlled-$R_y$ and $R_z$ rotations. 
No additional multi-controlled gates are required.

\begin{table}[htbp]
  \renewcommand{\arraystretch}{1.5}
  \centering
  \begin{threeparttable}
  \caption{Summary of Results on HW-$k$ State Preparation}
  \label{tab:summary}
  \begin{tabular}{cccc} 
    \toprule \toprule
    Reference & Circuit Size  & Ancillas   \\
    \midrule 
    \cite{raveh2024deterministic} & $\bo{\binom{n}{k}k}$  & 0 \\
    \cite{farias2025quantum} & $\bo{\binom{n}{k}k}$  & 0 \\
    \cite{mao2024towards} & $\bo{\binom{n}{k}\log n}$, $\om{\binom{n}{k}}$ & $\bo{n}$ \\
    \textbf{Ours} & $\ta{\binom{n}{k}}$ & $\max\cbra{0, n-3}$\\
    \bottomrule \bottomrule
  \end{tabular}
\end{threeparttable}
\end{table}

\subsection{Proof Techniques}

Our proof proceeds by constructing a \emph{Hamming tree} whose nodes are all associated with bitstrings of length $n$ and Hamming weight $k$, with the leaves collectively enumerating all such bitstrings.
Each node at level $i$ is associated with a string of the form $0^{\,n-i-\ell}1^{\,\ell}b$, 
where $b$ is a bitstring of length $i$ and $\ell$ is an integer such that $\ell + \mathrm{HW}(b) = k$. 
Intuitively, the suffix $b$ is gradually extended until the entire string of Hamming weight $k$ is specified. 
The two children of such a node are given by $0^{\,n-i-\ell-1}1^{\ell}0b$ and $0^{\,n-i-\ell}1^{\ell}b$, 
respectively. 
Although the right child shares the same representation as its parent, it can be viewed as advancing the construction by fixing one more bit in the suffix.

We traverse this Hamming tree in a \emph{preorder} manner, appending gates to the circuit at each internal node to realize the corresponding branching. 
To ensure the correct application of gates, we employ ancillary qubits to track the current position within the tree. 
Specifically, an ancillary qubit $a_i$ indicates whether the suffix $\ket{q_i q_{i+1}\dots q_n}$ of a basis state is in configuration $\ket{b}$. 
We first compute $a_{n-1}$ from $(q_n, q_{n-1})$, then successively compute $a_{n-2}$ from $(a_{n-1}, q_{n-2})$, and so on, each time using Toffoli gates and $X$ gates.
Note that the value of $a_{n-i}$ is updated only when visiting a node at level $i$; consequently, some ancillary qubits are refreshed frequently while others remain unchanged for long stretches of the traversal.
This mechanism ensures that the circuit size scales linearly with the size of the Hamming tree.


\section{Preliminaries}\label{section:pre}

\subsection{String Notations}

We denote a binary string of length $n$ by $x = x_1 x_2 \dots x_n \in \{0,1\}^n$.  
For indices $1 \le i \le j \le n$, we use $x[i,j]$ to represent the substring $x_i x_{i+1} \dots x_j$.  
Concatenation of two strings $x$ and $y$ is denoted by $xy$.  
Given integers $n$ and $k$, for a binary string $s$, we define
\begin{align}
    \mathcal{B}_s^{\,n,k} = \{ x \in \{0,1\}^n : \mathrm{HW}(x) = k \text{ and } x[n-|s|+1, n] = s \},
\end{align}
i.e., the set of all binary strings of length $n$ and Hamming weight $k$ whose last $|s|$ bits form the substring $s$.  
When $n$ and $k$ can be inferred from the context, we simply write $\mathcal{B}_s$.

\subsection{Basic Quantum Gates}

We use standard quantum gates as building blocks in our circuits, including $X$, CNOT, and Toffoli (CCX) gates.  
In addition, we frequently use single-qubit rotation gates \emph{U3 gate} parameterized by three angles $(\theta, \phi, \lambda)$:
\begin{align}
U3(\theta, \phi, \lambda) =
\begin{pmatrix}
\cos(\theta/2) & -e^{i \lambda}\sin(\theta/2) \\
e^{i \phi}\sin(\theta/2) & e^{i(\phi+\lambda)}\cos(\theta/2)
\end{pmatrix}.
\end{align}

It can be implemented as a sequence of rotations around the $Y$ and $Z$ axes:
\begin{align}
U3(\theta, \phi, \lambda) = R_z(\phi)\, R_y(\theta)\, R_z(\lambda),
\end{align}
where
\begin{align}
R_y(\theta) = \begin{pmatrix} \cos(\theta/2) & -\sin(\theta/2) \\ \sin(\theta/2) & \cos(\theta/2) \end{pmatrix}, \quad
R_z(\theta) = \begin{pmatrix} 1 & 0 \\ 0 & e^{i \theta} \end{pmatrix}.
\end{align}

These gates form the basic building blocks for the HW-$k$ state preparation circuits described in this work.

\subsection{Binary Tree Basics}

A \emph{binary tree} is a tree data structure in which each node has at most two children. 
Nodes with no children are called \emph{leaf nodes}, while nodes with at least one child are called \emph{internal nodes}. 
For each internal node, we refer to its two children as the \emph{left child} and \emph{right child}, if they exist.
In a binary tree, where every internal node has exactly two children, 
the number of internal nodes is one less than the number of leaf nodes. 
This property will be useful in analyzing the complexity of our HW-$k$ state preparation algorithm.

\section{Algorithm for Preparing Fixed-Hamming-Weight States}
In this section, we provide a detailed description of our HW-$k$ state preparation algorithm. 
Our approach begins by constructing a binary tree that encodes the process of generating all binary strings with Hamming weight $k$. The construction of Hamming trees is inspired by Ref.~\cite{raveh2024deterministic}.
We then perform a pre-order traversal of this tree, and during the traversal, the corresponding quantum circuit is generated.

In the constructed binary tree, we denote the $j$-th node at level $i$ by $v_{i,j}$. 
For example, $v_{0,0}$ represents the root node. 
Each node is associated with a binary string of Hamming weight $k$. 
The string associated with the node at level $i$ is always of the form $0^{\,n-i-\ell}1^{\,\ell}b$, 
where $b$ is a bitstring of length $i$ and the integer $\ell$ satisfies $\ell + \mathrm{HW}(b) = k$. 
Therefore, we have $\max\cbra{0, k-i} \leq \ell \leq \min\cbra{n-i, k}$. 
In particular, the root node $v_{0,0}$ corresponds to the initial string $0^{\,n-k}1^k$, where $b$ is the empty string.

The Hamming tree is recursively constructed starting from the root according to the following rules. 
The construction process essentially extends the suffix $b$ until a binary string of Hamming weight $k$ is determined. 
Suppose the current node is at level $i$, corresponding to the string $0^{\,n-i-\ell}1^{\,\ell}b$. 
If $\ell = 0$ or $\ell = n-i$, the string of this form and Hamming weight $k$ is unique, 
and therefore no children are added to this node. 
Otherwise, two child nodes are added: 
the left child corresponds to the string $0^{\,n-i-\ell-1}1^{\,\ell}0b$, 
and the right child corresponds to the string $0^{\,n-i-\ell-1}1^{\,\ell-1}1b$. 
The procedure is then applied recursively to the left child and the right child. For a node $v$ in the Hamming tree, we use the following notation:
\begin{itemize}
    \item $\text{String}(v)$ denotes the binary string associated with the node $v$;
    \item $\text{Layer}(v)$ denotes the level of the node in the tree, with the root being at layer $0$;
    \item $\text{LeftChild}(v)$ and $\text{RightChild}(v)$ denote the left and right children of node $v$, respectively.
\end{itemize}
The algorithm for constructing a Hamming tree is shown in \cref{alg:hamming_tree}.

It is straightforward to verify that the leaves of a Hamming tree encode all binary strings of Hamming weight $k$.
A Hamming tree with $n=4$ and $k=2$ is shown in \cref{fig:hammingtree}.

\begin{algorithm}[H]
\caption{Constructing a Hamming Tree for HW-$k$ Strings}
\label{alg:hamming_tree}
\begin{algorithmic}[1]
\REQUIRE Integers $n$ (string length) and $k$ (Hamming weight)
\ENSURE Root node of the Hamming tree
\STATE Create root node with $\text{String(root)} \xleftarrow{} 0^{n-k}1^{k}$ and $\text{Layer(root)} \xleftarrow{} 0$ \\
\STATE \textbf{function} BuildHammingTree(node $v$) \\
\STATE \quad $s \xleftarrow{} \text{String}(v)$, $i \xleftarrow{} \text{Layer}(v)$, $\ell \xleftarrow{} \mathrm{HW}(s[1, n-i])$  \\
\STATE \quad \textbf{if} $\ell = 0 \text{ or } \ell = n-i$ \textbf{then}
\STATE \quad\quad \textbf{return} 
\STATE \quad \textbf{else}
\STATE \quad\quad Create node $v_0$ with $\text{String}(v_0) \xleftarrow{} 0^{n-i-\ell-1}1^{\ell}0s[n-i+1,n]$ and $\text{Layer}(v_0) \xleftarrow{} i+1$  \\
\STATE \quad\quad Create node $v_1$ with $\text{String}(v_1) \xleftarrow{} 0^{n-i-\ell-1}1^{\ell}s[n-i+1,n]$ and $\text{Layer}(v_1) \xleftarrow{} i+1$  \\
\STATE \quad\quad $\text{LeftChild}(v) \xleftarrow{} \text{BuildHammingTree}(v_0)$ \\
\STATE \quad\quad $\text{RightChild}(v) \xleftarrow{} \text{BuildHammingTree}(v_0)$ \\
\STATE \quad\quad \textbf{return} $v$
\STATE \quad \textbf{end if} \\
\STATE \textbf{end function}\\
\STATE \textbf{return} \text{BuildHammingTree}(root)
\end{algorithmic}
\end{algorithm}

\begin{figure}[htbp]
    \centering
    \includegraphics[width=0.75\linewidth]{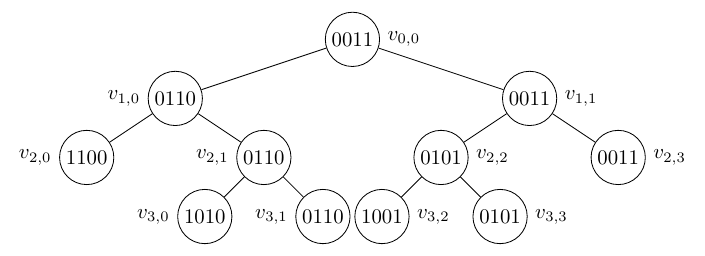}
    \caption{A Hamming tree with $n=4$ and $k=2$.}
    \label{fig:hammingtree}
\end{figure}

Before presenting our algorithm, we first introduce some notation. 
The working register consists of $n$ qubits $q_1, q_2, \dots, q_{n}$, and are used to encode binary strings. 
Specifically, for a binary string $s = s_1 s_2 \dots s_n$ of Hamming weight $k$, 
the bit $s_i$ is encoded in qubit $q_i$.
In addition to the working register, we also require an ancillary register consisting of $n-3$ qubits. 
For convenience, we denote these ancillary qubits by $a_3, a_4, \dots, a_{n-1}$. 
The layout of these qubits is illustrated in \cref{fig:layout} for the special case of $n=6$.

It is worth noting that the arrows in \cref{fig:layout} do not represent physical connections between qubits, 
but rather indicate the direction of information flow.
With a slight abuse of notation, throughout the paper the symbol of a qubit (e.g., $q_i$ or $a_j$)  
is sometimes also used to denote the classical bit value stored in it.
Specifically, $a_5$ stores the value $q_5 \wedge q_6$, which can be computed by applying a Toffoli gate with $q_5$ and $q_6$ as control qubits and $a_5$ as the target qubit.  
Similarly, $a_4$ stores $a_5 \wedge q_4$, $a_3$ stores $a_4 \wedge q_3$, and so on.
All the ancillary qubits are initially prepared in the state $\ket{0}$.  
If we first compute $a_{n-1}$, then $a_{n-2}$, and so on sequentially,  
the qubit $a_i$ ($3 \leq i \leq n-1$) will store the value  
$q_i \wedge q_{i+1} \wedge \cdots \wedge q_n$, which serves as an indicator of whether  
the string stored in $q_i q_{i+1} \cdots q_n$ is equal to $1^i$.  
More generally, to check whether the substring $q_i q_{i+1} \cdots q_n$ equals an arbitrary string $s \in \{0,1\}^i$,  
we may first apply $\text{X}$ gates to those qubits $q_j$ for which $s_j=0$,  
and then compute $a_{n-1}, a_{n-2}, \dots, a_i$ sequentially in the same way.

\begin{figure}[hptb]
    \centering
    \includegraphics[width=0.75\linewidth]{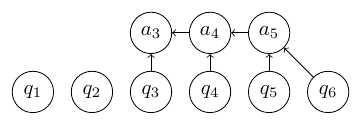}
    \caption{Layout of qubits with $n=6$.}
    \label{fig:layout}
\end{figure}

\begin{theorem}[Restatement of \cref{thm:hwk}]
    For any integers $n$ and $k$, any $n$-qubit HW-$k$ state can be prepared by a quantum circuit of size $\bo{\binom{n}{k}}$ using $\max\cbra{0, n-3}$ ancillary qubits.
\end{theorem}
\begin{proof}
The algorithm is shown in \cref{alg:hwk_synthesis}. 
First, we construct a Hamming tree for parameters $n$ and $k$, 
and denote its root node by ${root}$. 
We then initialize an empty quantum circuit $QC$ with $n + m$ qubits, 
and transform the initial state $\ket{0^{n+m}}$ into $\ket{0^{\,n-k} 1^k 0^m}$ 
using $k$ $X$ gates, where $m = \max\cbra{0, n-3}$. 
Here, the first $n$ qubits form the working register, and the last $m$ qubits form the ancillary register. 
Next, we define a recursive function $\texttt{StatePrep}$ and call it on the root node, $\texttt{StatePrep}({root})$. 
During the execution, quantum gates are appended to $QC$. 
After the recursive procedure finishes, the quantum circuit $QC$ preparing the HW-$k$ state is returned.


\paragraph{Description of \texttt{StatePrep}.}  
Suppose the current node is at depth $i$ in the Hamming tree.  
Let $s$ denote the string associated with this node and $\ell = \mathrm{HW}(s[1, n-i])$. Let $p_0,p_1, \theta, \phi, \lambda$ be defined as follows:
\begin{align}\label{equ:p0}
p_0 &= \begin{cases}
    \sqrt{\sum_{x \in \mathcal{B}_{0s[n-i+1, n]}} |\alpha_x|^2},& \text{if } \abs{\mathcal{B}_{0s[n-i+1, n]}} > 1, \\
    \alpha_{s'}, \text{where }s'\text{ is the unique string in }\mathcal{B}_{0s[n-i+1, n]}, &\text{if } \abs{\mathcal{B}_{0s[n-i+1, n]}} = 1.
\end{cases} \\
p_1 &= \begin{cases}
    \sqrt{\sum_{x \in \mathcal{B}_{1s[n-i+1, n]}} |\alpha_x|^2},&  \text{if } \abs{\mathcal{B}_{1s[n-i+1, n]}} > 1, \\
    \alpha_{s''}, \text{where }s''\text{ is the unique string in }\mathcal{B}_{1s[n-i+1, n]}, & \text{if } \abs{\mathcal{B}_{1s[n-i+1, n]}} = 1. \label{equ:p1}
\end{cases}
\end{align}
\begin{align}\label{equ:thetaphilamda}
    \theta = 2\arccos(\frac{\abs{p_1}}{\sqrt{\abs{p_0}^2 + \abs{p_1}^2}}), 
    \quad \phi = \arg(p_0)-\pi, 
    \quad \lambda = \arg(p_1)-\phi.
\end{align}

\begin{itemize}
    \item \textbf{Case $i=0$.}  
    The current node is the root.  
    We apply the circuit shown in \cref{fig:circuita}. 
    Specifically, we first apply a $U3(\theta, \phi, \lambda)$ gate on $q_{n}$. 
    This is followed by a circuit consisting of two X gates and a CNOT gate, which flips $q_{n-k}$ when $q_n$ is in state $\ket{0}$. 
    The action of this circuit is to transform the state $\ket{0^{n-k}1^k}$ into
    \begin{align}
    \label{equ:1}
        \sqrt{\sum_{x \in \mathcal{B}_{0}} \abs{\alpha_x}^2}\ket{0^{\,n-k-1}1^k 0} 
        + \sqrt{\sum_{x \in \mathcal{B}_{1}} \abs{\alpha_x}^2}\ket{0^{\,n-k}1^k}.
    \end{align}

    If the left child is an internal node, we first apply an $X$ gate on $q_n$, 
    recursively call $\texttt{StatePrep}$ on the left child, and finally apply an $X$ gate again to restore $q_n$.  
    Similarly, we recursively call $\texttt{StatePrep}$ on the right child if it is an internal node.
    
    \item \textbf{Case $i>0$.}  
    The current node corresponds to an internal node of the Hamming tree.
    For simplicity, we assume $i>2$; otherwise, one should replace $a_{n-i+1}$ by $q_n$ 
    (as in line~\ref{line2:23} of \cref{alg:hwk_synthesis}). 
    
    We assert that, in the current quantum state, 
    the ancillary qubit $a_{n-i+1}$ is in state $\ket{1}$ 
    if and only if the suffix of the work register matches $s[n-i+1, n]$.  
    Indeed, among all basis states, the only one satisfying this condition is $\ket{s}$. 
    This ensures that all other basis states remain unchanged when transforming $\ket{s}$.
    
    We then apply the circuit shown in \cref{fig:circuitb}. 
    Specifically, we first apply a controlled-$U3(\theta, \phi, \lambda)$ gate with $a_{n-i-1}$ as the control qubit and $q_{n-i}$ as the target qubit.
    This is followed by a circuit consisting of two X gates on $q_{n-i}$ and a Toffoli gate with $a_{n-i-1}$ and $q_{n-i}$ as control qubits 
    and $q_{n-i-\ell}$ as the target qubit, which flips $q_{n-i-\ell}$ when $a_{n-i-1}q_{n-i}$ is in state $\ket{10}$.
    The action of this circuit transforms the basis state 
        $\sqrt{\sum_{x\in \mathcal{B}_{s[n-i+1, n]}}\abs{\alpha_x}^2}\ket{s}$
    into
    \begin{align}
    \label{equ:2}
        \sqrt{\sum_{x\in \mathcal{B}_{0s[n-i+1, n]}}\abs{\alpha_x}^2}\ket{0^{n-i-\ell-1}1^{\ell}0\,s[n-i+1, n]} 
        + \sqrt{\sum_{x\in \mathcal{B}_{1s[n-i+1, n]}}\abs{\alpha_x}^2}\ket{0^{n-i-\ell}1^{\ell}s[n-i+1, n]}.
    \end{align}

    If the left child of the current node is an internal node, 
    we first apply an $X$ gate on $q_{n-i}$, 
    followed by a Toffoli gate controlled by $(a_{n-i+1}, q_{n-i})$ with target $a_{n-i}$.  
    We then recursively call $\texttt{StatePrep}$ on the left child, 
    and finally undo the Toffoli gate and the $X$ gate to restore $q_{n-i}$.  
    For the right child, we first apply the Toffoli gate, 
    then recursively call $\texttt{StatePrep}$, 
    and finally undo the Toffoli gate to restore $q_{n-i}$.

    It is straightforward to verify the aforementioned assertion that, at the beginning of visiting the current node, 
    the ancillary qubit $a_{n-i+1}$ is in state $\ket{1}$ 
    if and only if the suffix of the work register matches $s[n-i+1, n]$ and among all basis states in the current state, the only one satisfying this condition is $\ket{s}$.  
    This can be proved by induction. Assume that the assertion holds for the current node.  
    When visiting the left child, we first apply an $X$ gate and a Toffoli gate, 
    which sets $a_{n-i}$ to $\ket{1}$ only for basis states whose suffix is $0s[n-i+1, n]$, 
    while leaving $a_{n-i}$ in $\ket{0}$ for all other basis states.  
    Among the basis states in the current quantum state, the only one with this suffix is 
    $\ket{0^{\,n-i-\ell-1}1^{\,\ell}0\,s[n-i+1, n]}$,
    which was just produced when visiting the current node.  
    Similarly, the assertion can be verified for the right child, 
    showing that it holds throughout the recursive traversal.
\end{itemize}

According to \cref{equ:1} and \cref{equ:2}, when visiting the current node, 
the basis state corresponding to the current node's string is transformed 
into a superposition of the strings corresponding to its child nodes.  
After traversing all internal nodes, the quantum state becomes a superposition 
of all strings associated with the leaf nodes, i.e., all the string of Hamming weight $k$. Therefore, by appropriately choosing the rotation angles of the $U3$ gates, 
we can finally prepare the target HW-$k$ state.

Finally, we analyze the gate complexity of the algorithm. 
At each internal node, we apply either a $U3$ gate or a controlled-$U3$ gate, two $X$ gates, and a CNOT or Toffoli gate. 
If the left child is also an internal node, two additional $X$ gates and a Toffoli gate are applied; 
if the right child is an internal node, an additional Toffoli gate is applied. 
Since the Hamming tree has $\binom{n}{k}$ leaf nodes, the number of internal nodes is $\binom{n}{k}-1$. 
Therefore, the circuit size and depth of the constructed circuit are both $\bo{\binom{n}{k}}$. The classical runtime for executing the algorithm is also $\bo{\binom{n}{k}}$.
\end{proof}

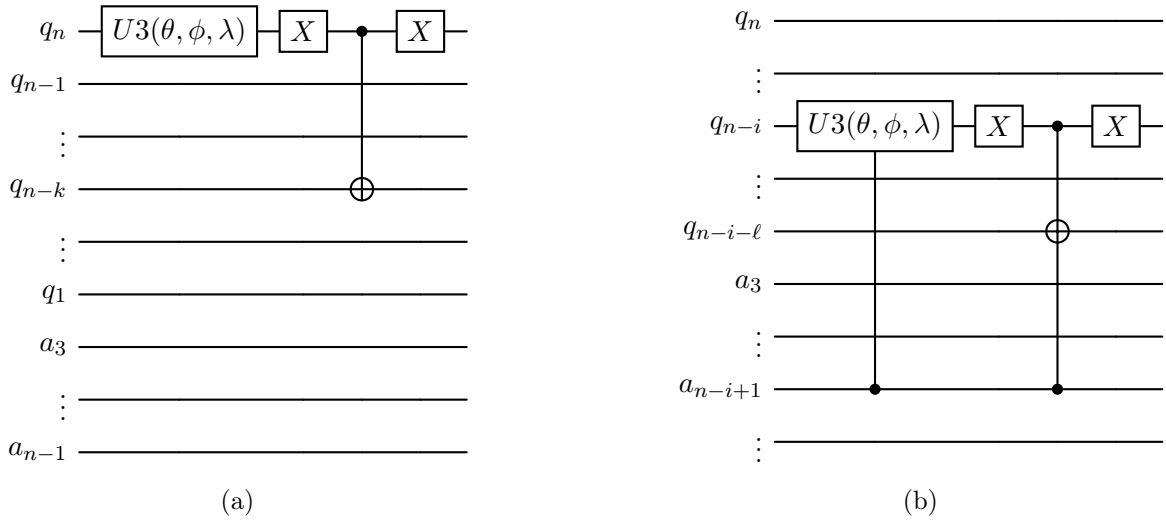
\begin{figure}[htbp]
    \centering
    \begin{subfigure}{0.45\textwidth}
        \centering
\begin{quantikz}[row sep={0.7cm,between origins},  column sep=0.3cm]
\lstick{$q_n$}       & \gate{U3(\theta, \phi, \lambda)} & \gate{X} & \ctrl{3} & \gate{X} &  \\
\lstick{$q_{n-1}$}   &   & & && \\
\lstick{\vdots}      & &  & &&\\
\lstick{$q_{n-k}$}   &  & &\targ{}& &  \\
\lstick{\vdots}      & & & & &\\
\lstick{$q_1$}       &  & & & &\\
\lstick{$a_3$}       &  & & & &\\
\lstick{\vdots}      &  & & & &\\
\lstick{$a_{n-1}$}   &  & & & &
\end{quantikz}
        \caption{}
        \label{fig:circuita}
    \end{subfigure}
    \hfill
    \begin{subfigure}{0.45\textwidth}
        \centering
\begin{quantikz}[row sep={0.7cm,between origins}, column sep=0.3cm]
\lstick{$q_n$}         &    &  & & &\\
\lstick{\vdots}        &    &     & & &    \\
\lstick{$q_{n-i}$}     &    \gate{U3(\theta, \phi, \lambda)}& \gate{X} & \ctrl{2} & \gate{X} &  \\
\lstick{\vdots}        &    &     & & &    \\
\lstick{$q_{n-i-\ell}$}&     && \targ{} & &  \\
\lstick{$a_{3}$}       &    &    & & &\\
\lstick{\vdots}        &    &     & & &\\
\lstick{$a_{n-i+1}$}   &   \ctrl{-5} && \ctrl{-3}  & &\\
\lstick{\vdots}        &    &     & & &
\end{quantikz}
        \caption{}
        \label{fig:circuitb}
    \end{subfigure}
    \caption{Circuit diagrams for \cref{alg:hwk_synthesis}.}
\end{figure}

\begin{algorithm}[htbp]
\caption{Synthesis of HW-$k$ State}
\label{alg:hwk_synthesis}
\begin{algorithmic}[1]
\REQUIRE Number of qubits $n$, Hamming weight $k$, amplitudes $\cbra{\alpha_x}_{x:\mathrm{HW}(x)=k}$
\ENSURE Quantum circuit preparing the HW-$k$ state
\STATE $m \gets \max\cbra{0, n-3}$ \quad // number of ancillary qubits \label{line2:1}
\STATE root $\xleftarrow{}$ construct a Hamming tree by \cref{alg:hamming_tree} with input $n$ and $k$;
\STATE Initialize a quantum circuit $QC$ with $n+m$ qubits
\FOR{$i = n-k+1$ to $n$}
    \STATE Apply $X$ gate to qubit $q_i$
\ENDFOR \label{line2:6}
\STATE\textbf{function} StatePrep(node $v$)
\STATE \quad \textbf{if} $v$ does not have child nodes \textbf{then} \textbf{return} \textbf{end if}
\STATE \quad $s \xleftarrow{} \text{String}(v)$, $i \xleftarrow{} \text{Layer}(v)$, $\ell \xleftarrow{} \mathrm{HW}(s[1,n-i])$ 
\STATE \quad Let $p_0, p_1, \theta, \phi, \lambda$ be defined as in \cref{equ:p0}, \cref{equ:p1} and \cref{equ:thetaphilamda}
\STATE \quad \textbf{if} i = 0 \textbf{then} \quad // Root node
\STATE \quad\quad Append to $QC$ a $U3(\theta, \phi, \lambda)$ gate on $q_{n}$
\STATE \quad\quad Append to $QC$ an $X$ gate on $q_{n}$, a CNOT gate with $q_{n}$ as control and $q_{n-k}$ as target, followed by another $X$ gate on $q_{n}$
\STATE \quad\quad \textbf{if} LeftChild($v$) is an internal node \textbf{then}
\STATE \quad\quad\quad Append to $QC$ a X gate on $q_{n}$ 
\STATE \quad\quad\quad StatePrep(LeftChild($v$))
\STATE \quad\quad\quad Append to $QC$ a X gate on $q_{n}$ 
\STATE \quad\quad \textbf{end if}
\STATE \quad\quad \textbf{if} RightChild($v$) is an internal node \textbf{then} 
\STATE \quad\quad\quad StatePrep(RightChild($v$))
\STATE \quad\quad \textbf{end if}
\STATE \quad\textbf{else}
\STATE \quad\quad \textbf{if} $i=1$ \textbf{then} $c\xleftarrow{}q_n$ \textbf{else} $c\xleftarrow{} a_{n-i+1}$ \textbf{end if} \label{line2:23}
\STATE \quad\quad Append to $QC$ a controlled-$U3(\theta, \phi, \lambda)$ gate with qubit $c$ as control and $q_{n-i}$ as target
\STATE \quad\quad Append to $QC$ an $X$ gate on $q_{n-i}$, a Toffoli gate with qubit $c$ and $q_{n-i}$ as control and  $q_{n-i-\ell}$ as target, followed by another $X$ gate on $q_{n}$
\STATE \quad\quad \textbf{if} LeftChild($v$) is an internal node is not empty \textbf{then}
\STATE \quad\quad\quad Append to $QC$ a X gate on $q_{n-i}$, followed by a Toffoli gate with qubit $c$ and $q_{n-i}$ as control and  $a_{n-i}$ as target
\STATE \quad\quad\quad StatePrep(LeftChild($v$))
\STATE \quad\quad\quad Append to $QC$ a Toffoli gate with qubit $c$ and $q_{n-i}$ as control and  $a_{n-i}$ as target, followed by a X gate on $q_{n-i}$
\STATE \quad\quad \textbf{end if}
\STATE \quad\quad \textbf{if} RightChild($v$) is an internal node \textbf{then} 
\STATE \quad\quad\quad Append to $QC$ a Toffoli gate with qubit $c$ and $q_{n-i}$ as control and  $a_{n-i}$ as target
\STATE \quad\quad\quad StatePrep(RightChild($v$))
\STATE \quad\quad\quad Append to $QC$ a Toffoli gate with qubit $c$ and $q_{n-i}$ as control and  $a_{n-i}$ as target
\STATE \quad\quad \textbf{end if}
\STATE \quad\textbf{end if}
\STATE \textbf{end function}
\STATE StatePrep(root)
\STATE \textbf{return} $QC$
\end{algorithmic}
\end{algorithm}

In \cref{fig:example_circuits}, we present example quantum circuits generated by our algorithm for small systems. 
\cref{fig:circuit_n3_k1} shows the circuit for $n=3$ and $k=1$, while \cref{fig:circuit_n4_k1} shows the circuit for $n=4$ and $k=1$. 
These examples demonstrate that the circuits generated by our algorithm are highly efficient.
We note that there are consecutive $X$ gates in the circuits. 
These can be removed through further optimization; 
however, we keep them here to help the reader understand the circuits in accordance with \cref{alg:hwk_synthesis}.

\begin{figure}[htbp]
    \centering
    \begin{subfigure}[b]{0.45\textwidth}
        \centering
        \includegraphics[width=\textwidth]{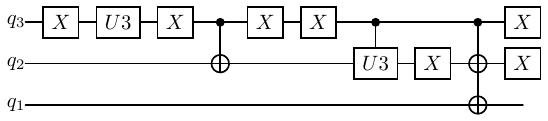} 
        \caption{Circuit for $n=3, k=1$}
        \label{fig:circuit_n3_k1}
    \end{subfigure}
    \hfill
    \begin{subfigure}[b]{0.45\textwidth}
        \centering
        \includegraphics[width=\textwidth]{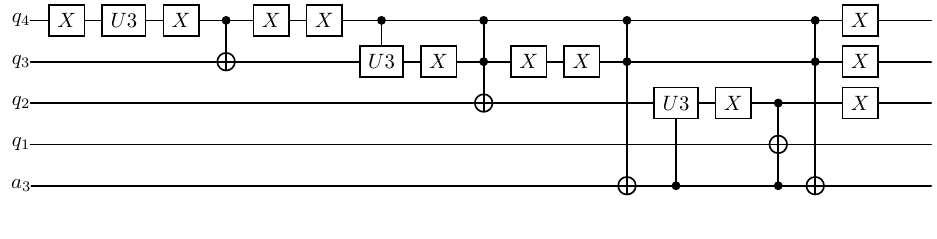} 
        \caption{Circuit for $n=4, k=1$}
        \label{fig:circuit_n4_k1}
    \end{subfigure}
    \caption{Examples of quantum circuits for preparing HW-$k$ states.}
    \label{fig:example_circuits}
\end{figure}

The method of employing ancillary qubits in this work can be extended to other algorithms. 
For instance, in Ref.~\cite{mozafari2022efficient}, an Algebraic Decision Diagram (ADD) is used to represent a quantum state, 
and a state-preparation circuit is constructed by traversing the ADD. 
That approach achieves a circuit size of $O(kn)$, where $k$ is the number of paths in the ADD and $n$ is the number of qubits, 
while requiring only one ancillary qubit. 
In contrast, by employing up to $n$ ancillary qubits to keep track of the current path, 
we can reduce the circuit size to $O(|V|)$, where $|V|$ denotes the number of nodes in the ADD.
Note that $O(|V|)$ is strictly smaller than $O(kn)$, especially when the ADD contains a large number of branching nodes.

\section{Conclusion}\label{section:conclusion}
We presented an algorithm for preparing HW-$k$ states with a circuit size of $\bo{\binom{n}{k}}$, matching the theoretical lower bound, 
while using at most $\max\{0, n-3\}$ ancillary qubits. 
Our construction improves upon previous methods in both gate-count and ancillary qubit number, 
employing only $X$, CNOT, Toffoli, and single- or controlled-$R_y$ and $R_z$ gates. 
The approach may be extended to other decision-diagram-based state preparation algorithms, 
potentially reducing gate complexity and ancillary qubit requirements.


Note that while we were preparing this manuscript, Li \emph{et al.}~\cite{li2025} independently released a preprint proposing an algorithm for preparing HW-$k$ states.  Their construction achieves a circuit with $O(\binom{n}{k})$-size and $O(\log \binom{n}{k})$-depth.  Logarithmic depth is interesting, but comes at the cost of  $O(\binom{n}{k})$ ancillary qubits. When $k \ll n$, the number of ancillary qubits amounts to roughly $n^k/k!$, which grows almost exponentially with $k$. Our construction   focuses on the optimal circuit size without optimizing depth, and only requires at most $n-3$ ancillary qubits, greatly alleviating the demand for ancillary qubits. A problem worth studying in the future is the trade-off between depth and the number of ancillary qubits. In addition, the methods used in Ref.~\cite{li2025} and here are fundamentally different: their method first loads amplitudes using the unary encoding, which already incurs $\binom{n}{k}$ ancillary qubits before converting to the binary encoding, whereas our method relies on traversing a Hamming tree with ancilla-assisted path tracking.


\bibliographystyle{unsrt} 
\bibliography{refs}


\end{document}